\documentclass[aps,pra,superscriptaddress,twocolumn,showpacs]{revtex4}
\usepackage[colorlinks, linkcolor=blue, citecolor=blue, urlcolor=blue, breaklinks=true]{hyperref}
\usepackage{amsmath}
\usepackage{graphicx}
\usepackage{amssymb}
\usepackage{bbm, color}
\usepackage{relsize}
\usepackage[title]{appendix}
\usepackage{amsthm}
\usepackage{lipsum}
\usepackage{listings}
\usepackage{bbold}

\theoremstyle{plain}
\newtheorem{thm}{Theorem}
\newtheorem{lem}[thm]{Lemma}

\theoremstyle{plain}
\newtheorem{defn}[thm]{Definition}

\newcommand{\ket}[1]{\left|#1\right>}

\newcommand*{\tn}[1]{{\textnormal{#1}}}

\newcommand{\xx}{\mathcal{X}}
\newcommand{\ma}{\mathcal{A}}
\newcommand{\mb}{\mathcal{B}}
\newcommand{\ba}{\begin{eqnarray}}
\newcommand{\ea}{\end{eqnarray}}

\newcommand{\dm}[1]{\ketbra{#1}{#1}}
\newcommand{\ketbra}[2]{|{#1}\rangle\!\langle{#2}|}
\newcommand{\eq}[1]{(\hyperref[eq:#1]{\ref*{eq:#1}})}
\newcommand{\tI}{\tilde{I}}
\newcommand{\tD}{\tilde{D}}
\newcommand{\mr}{\mathcal{R}}
\newcommand{\lemm}[1]{\hyperref[lemm:#1]{Lemma~\ref*{lemm:#1}}}

\DeclareMathOperator{\Tr}{Tr}

\begin{document}

\title{Activation and superactivation of single-mode Gaussian quantum channels }
\author{Youngrong Lim}
\affiliation{School of Computational Sciences, Korea Institute for Advanced Study, Seoul 02455, Korea}
\affiliation{Research Institute of Mathematics, Seoul National University, Seoul 08826, Korea}
\affiliation{Department of Mathematics and Research Institute for Basic Sciences, Kyung Hee University, Seoul 02447, Korea}
 \author{Ryuji Takagi}
 \affiliation{Center for Theoretical Physics and Department of Physics, Massachusetts Institute of Technology, Cambridge, Massachusetts 02139, USA}
 \author{Gerardo Adesso}
 \affiliation{Centre for the Mathematics and Theoretical Physics of Quantum Non-Equilibrium Systems, School of Mathematical Sciences, University of Nottingham, University Park, Nottingham NG7 2RD, United Kingdom}
\author{Soojoon Lee}
\affiliation{Department of Mathematics and Research Institute for Basic Sciences, Kyung Hee University, Seoul 02447, Korea}
\affiliation{Centre for the Mathematics and Theoretical Physics of Quantum Non-Equilibrium Systems, School of Mathematical Sciences, University of Nottingham, University Park, Nottingham NG7 2RD, United Kingdom}

\date{\today}
\pacs{03.67.-a, 03.65.Ud, 03.67.Bg, 42.50.-p}

\begin{abstract}
Activation of quantum capacity is a surprising phenomenon according to which the quantum capacity of a certain channel may increase by combining it with another channel with zero quantum capacity. Superactivation describes an even more particular occurrence, in which both channels have zero quantum capacity, but their composition has a nonvanishing one. We investigate these effects for all single-mode phase-insensitive Gaussian channels, which include thermal attenuators and amplifiers, assisted by a two-mode positive-partial-transpose channel. Our result shows that activation phenomena are special but not uncommon. We can reveal superactivation in a broad range of thermal attenuator channels, even when the transmissivity is quite low. This means that we can transmit quantum information reliably through very noisy Gaussian channels having zero quantum capacity. We further show that no superactivation is possible for entanglement-breaking Gaussian channels in physically relevant circumstances by proving the non-activation property of the coherent information of bosonic entanglement-breaking channels with finite input energy.
 \end{abstract}

\maketitle

\section{Introduction} \label{intro}
Quantum channels are ubiquitous tools for quantum information theory, quantum communication, and open quantum dynamics. The capacity of a channel is a central metric to assess its capability of reliably transmitting information over a large number of uses with asymptotically vanishing error.  There are several relevant notions of channel capacity depending on the given physical setting and type of information to be sent. For instance, the classical capacity is the transmission rate at which classical bits can be reliably sent~\cite{Holevo} while the quantum capacity refers to the corresponding quantity when quantum bits are to be sent~\cite{Shor95}.
The private capacity is another relevant quantity that plays a central role in cryptographical settings where one is to send classical bits with privacy~\cite{Private}.

Unfortunately, explicit formulas of the channel capacities have been known only for restricted cases. The reason is that, in general, nontrivial regularization formulas are  needed to characterize channel capacities. In other words, additivity no longer holds in general for one-shot capacity functions. This additivity violation has been proved for classical capacity~\cite{Hastings}, private capacity~\cite{Smith08,Li} and quantum capacity~\cite{DiVin,Smith07}. In particular, a stronger superadditive effect exists for the quantum capacity, called {\em superactivation}, in which we can have a positive quantum capacity for the product of two channels, even though each channel has zero quantum capacity on its own~\cite{Science}. Superactivation has also been found to occur in special instances of Gaussian channels~\cite{Smith}. This is an important observation, because Gaussian channels and Gaussian systems are implementable by simple quantum optical instruments~\cite{RMP}, e.g., phase shifters, beam splitters, single- and two-mode squeezers, and describe information transmission over optical fibres and real world telecommunications.

In the original work~\cite{Smith}, the two Gaussian channels for demonstrating superactivation were identified as the single-mode quantum-limited attenuator corresponding to the 50/50 beamsplitter, and a specific form of two-mode positive-partial-transpose (PPT) channel. Recently, {\em activation} effects (i.e., the fact that the quantum capacity of a channel is increased by combining it with a zero capacity channel) have been observed for Gaussian lossy channels corresponding to beamsplitters with a wider range of transmissivity.~\cite{Lim18}.

Here, we perform a systematic analysis of activation and superactivation effects in all single-mode phase-insensitive Gaussian channels, encompassing thermal attenuators and amplifiers, which model many physical situations and optical communication schemes~\cite{RMP,book,Holevobook,Caves}.  We show in particular that (super)activation is possible in a broad range of parameters for thermal attenuators, even when the corresponding beamsplitter transmissivity is quite low ($<$0.2). These are very noisy channels in the sense that only a small portion of the input state can be transmitted through them. Since the thermal attenuators for which the superactivation effect is confirmed are close to the entanglement-breaking (EB) channels~\cite{EB}, we also address the question whether it is possible to observe the same effect for EB channels. EB channels always have zero quantum capacity due to their anti-degradable property~\cite{Wildebook}, and it is known that EB channels with finite-dimensional input and output spaces cannot be superactivated~\cite{Watrous, Elton} (See also Appendix B). We extend this no-go result to infinite-dimensional bosonic EB channels with finite input energy, which implies that EB channels cannot be helped by another zero-capacity channel for transmitting quantum information in physically relevant circumstances.

 In Section~\ref{review}, some basic definitions and relations related to our work are introduced. In Section~\ref{main}, the main results are presented with some numerical and analytical methods. Finally, in Section~\ref{discussion}, we comment on a few remarks and open problems.

 \section{Preliminaries} \label{review}
Let us consider an isometry $V : {\cal H}(A) \rightarrow {\cal H}(B) \otimes {\cal H}(E)$. A quantum channel $\Phi : \rho_A \mapsto \rho_B$ is a completely-positive trace-preserving (CPTP) map corresponding to the action of the isometry on the input state of system $A$ followed by tracing out the environment $E$, written as $\Phi(\rho_A)=\text{Tr}_E V \rho_A V^{\dag}$~\cite{Stine}. If we trace out the output system $B$ instead of the environment, we get the complementary channel such as  $\Phi^c(\rho_A)=\text{Tr}_B V \rho_A V^{\dag}$. The quantum capacity $Q(\Phi)$ is defined as the maximum transmission rate of qubits through a given channel $\Phi$ with asymptotically vanishing error. By the quantum capacity theorem~\cite{Devetak05,Hayden08}, it is related to an entropic quantity called the coherent information, given by
\begin{equation}\label{coh}
I_c(\Phi,\rho)=H(\Phi(\rho))-H(\Phi^c(\rho)),
\end{equation}
where $H$ is the von Neumann entropy and $\rho$ is an input state of the channel. Then, the quantum capacity is given by
\begin{equation}
Q(\Phi)=\lim_{n\rightarrow \infty} \sup_{\rho_n} {I_c(\Phi^{\otimes n},\rho_n) \over n},
\end{equation}
where $\Phi^{\otimes n}$ means $n$ independent parallel uses of the channel, and $\rho_n$ is any state acting on ${\cal H}(A)^{\otimes n}$.

Gaussian states are the quantum states whose characteristic functions (or, equivalently, Wigner functions) have Gaussian distributions~\cite{Adesso,Serafini}. For an $n$-mode bosonic quantum state, there are $n$ pairs of position and momentum operators collectively written as $R=(Q_1,P_1,...,Q_n,P_n)^T$, that satisfy the commutation relation $[R_i,R_j]=iJ_{ij}$, where $J=\begin{pmatrix}0 &1 \\ -1  & 0\end{pmatrix}^{\oplus n}$. A Gaussian state can be entirely specified by the first and second moments of the quadrature operators instead of the density matrix $\rho$ itself, i.e., the displacement vector $d=\left< R \right>_{\rho}$, and the covariance matrix $\gamma$ with elements $\gamma_{ij}=\left<R_i R_j+R_j R_i \right>_{\rho}-2 \left<R_i \right>_{\rho} \left< R_j \right>_{\rho}$, respectively.

We focus our attention to Gaussian transformations, in which the quadrature operators are transformed by matrices in the real symplectic group, i.e., $S \in \text{Sp}(2n,\mathbb{R})$, $SJS^T=J$, such as $R^{\prime}=SR$. For each symplectic transformation $S$, there is a corresponding unitary transformation $U_S$, called symplectic unitary matrix, acting on quadrature operators as $R^{\prime}_i=U_S^{\dag} R_i U_S$ for $i=1,\dots,2n$. Then, a Gaussian channel is a CPTP map transforming Gaussian states to Gaussian states, which can be given by the symplectic dilation form as~\cite{Caruso}
\begin{equation}
\Phi_G(\rho_A)=\text{Tr}_E[U_S(\rho_A \otimes \rho_E)U^{\dag}_S],
\end{equation}
where $\rho_A$ is an input state and $\rho_E$ is a Gaussian state in the environment. In phase space, on the level of the covariance matrix $\gamma$ of a Gaussian state $\rho_A$, the action of a Gaussian channel can be expressed as $\gamma \rightarrow \Phi(\gamma) = X \gamma X^T +Y$, where $X$ and $Y=Y^T$ are $2n \times 2n$ real matrices constrained to the condition $Y+i(J-XJX^T)\geq 0$ to ensure that the channel is CPTP. In order to obtain the expression of the complementary channel, we need to consider a symplectic transformation having block matrix form $S=\begin{pmatrix} X &Z \\ X_c &Z_c \end{pmatrix}$. The number of modes of the input and output states is the same for the channels we care about in this work. If the environment modes are in vacuum states, a Gaussian channel and its complementary channel are described as $\Phi(\gamma)=X \gamma X^T +ZZ^T,~\Phi^c(\gamma)=X_c \gamma X^T_c +Z_c Z_c^T$.

For single-mode Gaussian channels, there exists a full classification~\cite{Holevo07}. Among those, we focus on the phase-insensitive channels, satisfying the condition that $X$ and $Y$ are diagonal. This class includes thermal attenuator, amplifier, and additive Gaussian noise channels. Note that the thermal attenuator is nothing but a beamsplitter operation with a transmissivity $t$ acting on the system mode $A$ and an ancillary environment mode $E$, after tracing out the latter. In general, the ancillary input of the beamsplitter can be in a thermal state with average photon number $N$. When the ancilla is in the vacuum state ($N=0$), the corresponding channel is known as quantum-limited attenuator. On the other hand, an amplifier channel corresponds to the operation consisting of a two-mode squeezer and a beam splitter on $A$ and $E$, which enables amplification of the input signal mode $A$. Similarly, if the environment mode $E$ is in the vacuum, we get a quantum-limited amplifier.

An EB channel always gives a separable output state, i.e., $\Phi \otimes \mathbb{1} (\rho_{AA^{\prime}})$ is separable, and it has zero quantum capacity. Similarly, an entanglement-binding channel, a type of PPT channel which also has zero quantum capacity, gives a non-distillable output state. In the Gaussian regime, because there is no bound entangled state of $1 \oplus n$ modes~\cite{bound}, an entanglement-binding channel needs at least a two-mode input and a two-mode output system. That is exactly the case for the PPT entanglement-binding channel that will be used in this work, suggested by~\cite{Smith}.
\begin{figure}[!t]
\includegraphics[width=8.5cm]{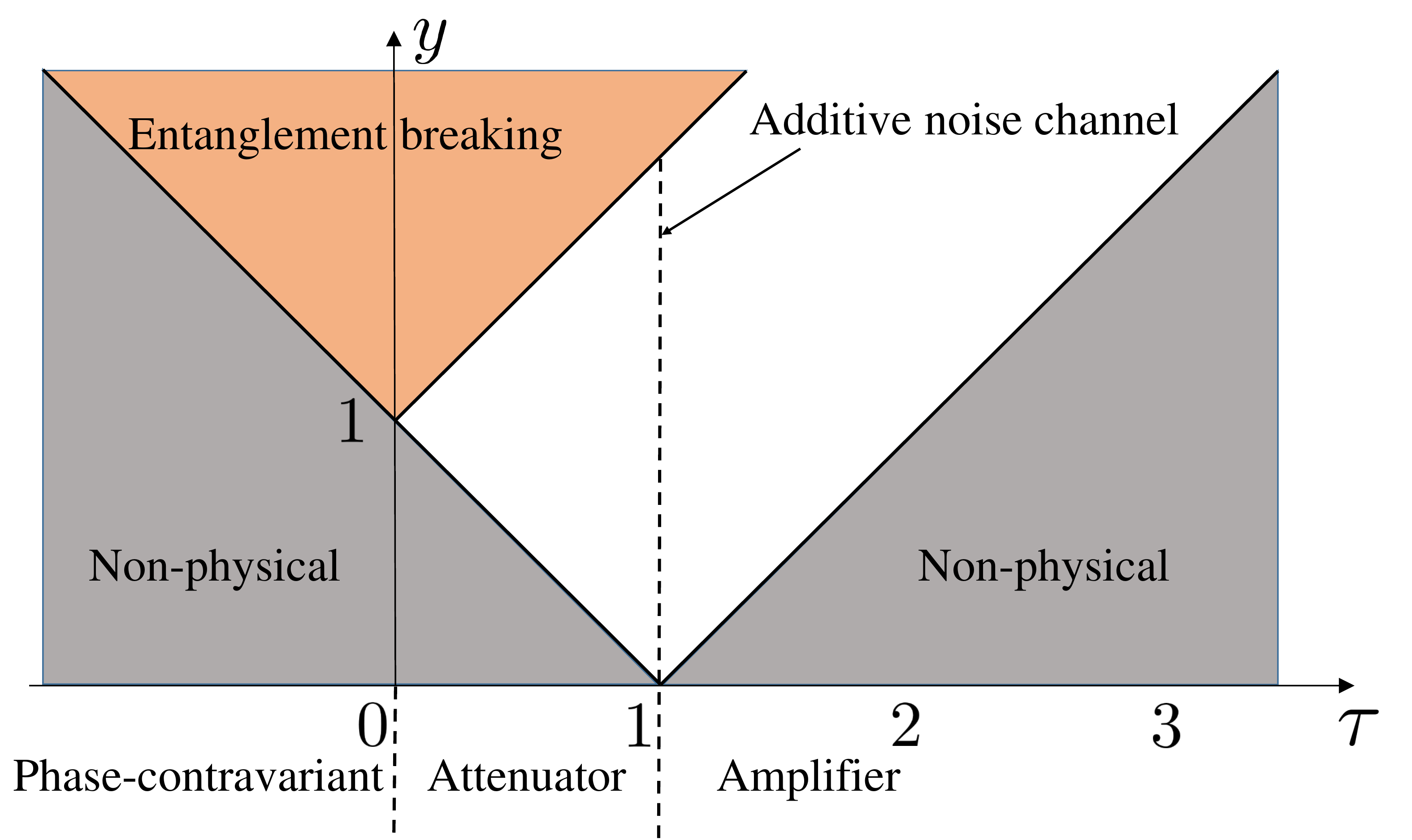}
\caption{(Color online) Classification of phase-insensitive single-mode Gaussian channels. Axes defined by $\tau=\det X, y=\sqrt{\det Y}$. Physical channels (CPTP) should satisfy the relation $y \ge |\tau-1|$. EB channels are on the (orange online) region of $y \ge |\tau|+1$.}
\label{fig1}
\end{figure}

\section{main results} \label{main}

We investigate which phase-insensitive single-mode Gaussian channels exhibit (super)activation of quantum capacity when combined with the two-mode PPT channel introduced in~\cite{Smith}. Our analysis will extend beyond the specific cases of the Gaussian lossy channel and the thermal attenuator with transmissivity near 0.5~\cite{Smith,Lim18}.

On the level of density matrices, a phase-insensitive channel $\Phi$ satisfies the condition
\begin{equation}
\Phi[e^{i\phi n_A}\rho e^{-i\phi n_A}]=e^{i\phi n_B}\Phi[\rho]e^{-i\phi n_B},
\end{equation}
where $\phi$ is any real number and $n_A$ ($n_B$) is the number operator on mode A (mode B).
As previously mentioned, phase-insensitive Gaussian channels are specified in phase space by diagonal matrices $X$ and $Y$.
All single-mode phase-insensitive Gaussian channels are depicted in Fig.~\ref{fig1} as a function of  $\tau=\det X$ and $y=\sqrt{\det Y}$, with $y \geq |\tau-1|$.

Let us consider the coherent information of the thermal attenuator $\Phi_{t,N}$, i.e., of the channel with $X=\sqrt{t}\mathbb{1}, Y=(1-t)(2N+1)\mathbb{1}$, where $0<t \le 1$ is the transmissivity and $N \geq 0$ is the mean photon number of the thermal noise. However, we cannot use the simple symplectic dilation explained in Section~\ref{review} because the thermal environment state is a mixed state. We can instead consider a symplectic dilation after purifying such thermal state to a pure two-mode squeezed state (Appendix A) to get the expression of the complementary channel. Apart from the case of zero thermal noise (equivalent to the quantum-limited attenuator, i.e., the Gaussian lossy channel), the exact formula for quantum capacity of the thermal attenuator is not known. However, there have been known not only lower bounds using a kind of thermal state input~\cite{Holevo01,nohlower}, but also the currently best upper bound as~\cite{upper,noh,plob},
\begin{align}
&{\cal{Q}}(\Phi_{t,N})\leq  \min \left\{ {\cal{Q}}_\text{data}(\Phi_{t,N}), {\cal{Q}}_\text{PLOB}(\Phi_{t,N})\right\} :={\cal{Q}}_U(\Phi_{t,N}),\nonumber \\
&{\cal{Q}}_\text{data}(\Phi_{t,N})=\tn{max}\left\{0,\tn{log}_2\left[ \frac{N(1-t)-t}{(1+N)(t-1)}\right] \right\}, \nonumber \\
&{\cal{Q}}_\text{PLOB}(\Phi_{t,N})=\tn{max}\left\{0,-\tn{log}_2 [(1-t)t^N]-g(N)\right\},
\end{align}
where $g(x)=(1+x)\log_2(1+x)-x\log_2 x$.

\begin{figure}[!t]
\includegraphics[width=8.5cm]{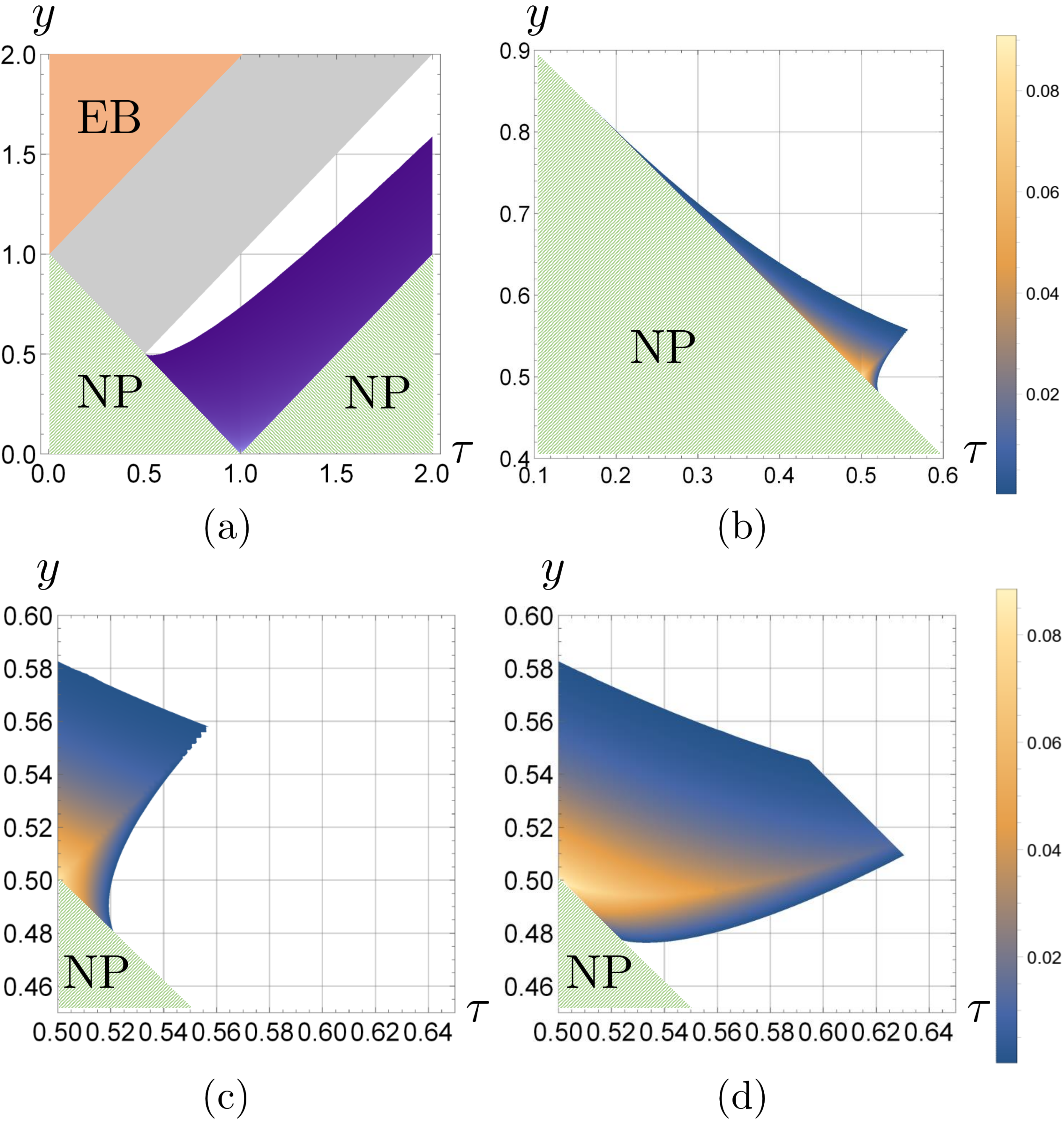}
\caption{(Color online) (a) Various regions having zero capacity (gray), zero maximum coherent information (white), and positive maximum coherent information (purple online, darker shading). EB region (orange online) and NP (non-physical) region (green online) are also specified~\cite{Kamil}. (b) Difference between the coherent information of the combined channel and the upper bound of quantum capacity. (c), (d) Comparing difference between the coherent information of the combined channel and the upper bound of quantum capacity/maximum coherent information in the region $\tau>0.5$. See text for further details.}
\label{fig2}
\end{figure}

We now have all the ingredients to test (super)activation of the quantum capacity. By using the symplectic dilation for thermal noise channels (Appendix A), we can obtain the covariance matrices of the combined channel output and complementary channel output. Since the PPT channel has zero quantum capacity, the coherent information of the combined channel should satisfy the following relation if there is no activation,
\begin{equation}\label{activation}
I_c(\Phi_\text{PPT} \otimes \Phi_{t,N},\rho_\text{in})\le {\cal{Q}}(\Phi_\text{PPT} \otimes \Phi_{t,N})\le {\cal{Q}}_U(\Phi_{t,N}),
\end{equation}
where $\Phi_\text{PPT}$ is a specific two-mode PPT channel suggested by Smith {\it et el.}~\cite{Smith}. Therefore, if we find an input state such that the coherent information of the combined channel exceeds the upper bound of the quantum capacity for the thermal attenuator, (super)activation is confirmed. In general, we need to search all possible three-mode input states, whose covariance matrices are described by 12 independent parameters, satisfying the physicality condition, i.e., $\gamma+i J \geq 0$~\cite{thmode}. Since the optimization over all those parameters is computationally intractable, we focus on a class of asymmetric input states specified by three parameters [Eq.~(\ref{input}) in Appendix A], generalizing a two-parameter family of input states used in previous works~\cite{Smith,Lim18}.

Although the quantum capacity of arbitrary single-mode phase-insensitive Gaussian channels is still unknown, there are more known facts regarding the maximal coherent information (one-shot quantum capacity)~\cite{Kamil}. In Fig.~\ref{fig2} (a), the gray region indicates channels with zero quantum capacity owing to their antidegradability, and the dark purple region contains channels with positive coherent information, thus also with positive quantum capacity. The intermediate (white) region, in between the purple and the gray regions, accommodates channels with zero maximum coherent information, but for which one cannot rule out the possibility of having positive quantum capacity. 

We compute numerically the difference between the coherent information with three-parameter optimized inputs of the combined channel, and the upper bound of the quantum capacity, i.e., $I_c(\Phi_\text{PPT} \otimes \Phi_{t,N},\rho_\text{in})-{\cal{Q}}_U(\Phi_{t,N})$, as in Fig.~\ref{fig2} (b). Our results show that (super)activation occurs in a broad range of parameters, even when the transmissivity is quite low ($\tau <0.2$). This result, which significantly extends previous findings \cite{Smith,Lim18}, also raises a question whether the violation of Eq.~(\ref{activation}) could be observed by a more thorough search when $\tau \to 0$ or even in the EB region.
For EB channels, however, we give a proof that it is not the case as long as the input states have finite energy (Appendix B). Further, we can show that our result covers all the three regions in Fig.~(\ref{fig2}) (a). Thus, there is supereactivation of quantum capacity and maximum coherent information for the gray regions.  Also, for the white region, there is superactivation of the maximum coherent information, as well as (super)activation of the quantum capacity. Finally, for the purple region, there is activation of the quantum capacity and maximum coherent information.  In addition, Fig.~(\ref{fig2}) (d) depicts the difference from the maximum coherent information instead of the upper bound for the quantum capacity. As expected, the region of activation of the maximum coherent information is much wider than the region of activation of the quantum capacity and the former fully incorporates the latter.

Another important remark is that in the $\tau>0.5$ region, we see that activation effects occur for thermal noise channels rather than quantum-limited channels (boundary on the non-physical channels) with the same transmissivity. For example, we cannot see any activation at $(\tau,y)=(0.53,0.47)$, but we see it at $(0.53,0.55)$. This seems counterintuitive, since thermal noise usually degrades the capacity of the channel, which means that it might prevent the activation. Because this can be a consequence of the fact that we have only constrained the optimization to a restricted family of input states, further investigation is needed to confirm these observations. We have also sought (super)activation for amplifier channels, but we cannot see any by our methods. This might come from the fact that the maximum coherent information has a relatively high value for the amplifiers, so it may limit activation. Therefore, we suggest a conjecture that single-mode Gaussian amplifiers cannot be activated.

\section{discussion}\label{discussion}

In this work, we have investigated the (super)activation of the quantum capacity in single-mode phase-insensitive Gaussian channels assisted with a two-mode positive-partial transpose channel. We found that, quite remarkably, a wide region of thermal attenuator channels can be activated, even when the transmissivity is quite low. 
This significantly extends the activatable region observed in the previous work, and our result gives a hope to further enlarge it by extending the search for the input space.
From our study, we cannot draw a conclusion about whether (super)activation happens also for the additive noise channels and amplifiers, but we conjecture these channels cannot exhibit (super)activation. 
% However, this may also be a drawback of the fact that we have not covered the full range of input states in our optimization, thus our result gives a hope to enlarge the activatable region by extending the search for the input space.

One can ask several questions about the (super)activation in Gaussian channels. First thing is finding tighter upper bounds of the quantum capacity for the amplifiers and the additive noise channels in order to test the activation conclusively. Second one is investigating multi-mode channels instead of single-mode ones. It could possibly give more classes having zero capacity or upper bounds on them. Finally, one could consider a single-mode phase-sensitive channel, which involves squeezing elements and is thus more complicated to handle. It has been known that for the standard method dealing with a PPT channel and an antidegradable channel, squeezing is needed for superactivation~\cite{Wolf13}. Thus, if we find other classes of channels having zero capacity, it could be superactivated in other ways without squeezing elements.

Our results show overall that quantum information can be transmitted reliably through a significant variety of thermal attenuator Gaussian channels, even when they are very noisy, when combined with other zero-capacity channels. This can be of practical relevance to extend the range and robustness of secure quantum communication with continuous variables.

\section*{ACKNOWLEDGMENTS}
We thank Elton Yechao Zhu for useful discussions. This work was supported by Basic Science Research Program through the National Research Foundation of Korea (NRF) funded by the Ministry of Science and ICT (NRF-2016R1A2B4014928 \& NRF-2017R1E1A1A03070510), the Ministry of Science and ICT, Korea, under an ITRC Program, IITP-2018-2018-0-01402 and Ministry of Education (NRF-2017R1A6A3A01007264 \& NRF-2018R1D1A1B07047512). R.T. acknowledges the support from NSF, ARO, IARPA, and the Takenaka Scholarship Foundation. G.A. acknowledges financial support from the European Research Council (ERC) under the Standard Grant GQCOP (Grant Agreement no. 637352).

\onecolumngrid

\begin{appendices}
\numberwithin{equation}{section}
\label{appA}
\section{Symplectic dilation for thermal environment}

If the environment (mode $E'$) is not a pure state, which corresponds to single-mode thermal attenuator/amplifier with $N\neq 0$, we need to find a symplectic transformation in order to get the expression for the complementary channel. In our cases, environment is a thermal state instead of vacuum state, having an average photon number $N$. Its covariance matrix is $\gamma_{\text{th}}=(2N+1)\mathbb{1}_2$. In this simple case, we can easily consider the purification for the thermal state and finally get a two-mode squeezed vacuum (TMSV) state, its covariance matrix is given by
\begin{equation}
\gamma_{\text{TMSV}}=\begin{pmatrix} (2N+1)\mathbb{1} & 2\sqrt{N(N+1)}\mathbb{Z} \\ 2\sqrt{N(N+1)}\mathbb{Z} & (2N+1)\mathbb{1} \end{pmatrix},
\end{equation}
where $\mathbb{Z}=\begin{pmatrix} 1 & 0 \\ 0 & -1 \end{pmatrix}$. The TMSV state is indeed a pure state because its symplectic
eigenvalues are 1's.

Now, we can write the symplectic transformation for the thermal attenuator with transmissivity $t$. For $N=0$, we know the symplectic transformation is written as
\begin{equation}
S_0=\begin{pmatrix} \sqrt{t}\mathbb{1} & \sqrt{1-t}\mathbb{1} \\ \sqrt{1-t}\mathbb{1} & -\sqrt{t}\mathbb{1} \end{pmatrix}.
\end{equation}

Let us set $X_0=\sqrt{t}\mathbb{1},~Z_0= \sqrt{1-t}\mathbb{1},~X_{c0}=\sqrt{1-t}\mathbb{1},~Z_{c0}=-\sqrt{t}\mathbb{1}$. Then we can find a symplectic transformation for a thermal attenuator with $N \neq 0$ such as
\begin{equation}
S_{\text{th}}=\begin{pmatrix} X_{\text{th}} & Z_{\text{th}} \\ X_{c,\text{th}} & Z_{c,\text{th}} \end{pmatrix}=\begin{pmatrix} X_0 & Z_0 & 0 \\ X_{c0} & Z_{c0} & 0 \\ 0 & 0 & \mathbb{1} \end{pmatrix},
\end{equation}
where $X_{\text{th}}=X_0,~Z_{\text{th}}=\begin{pmatrix}Z_0 & 0\end{pmatrix},~X_{c,\text{th}}=\begin{pmatrix} X_{c0} \\ 0 \end{pmatrix},~Z_{c,\text{th}}=\begin{pmatrix} Z_{c0} & 0\\ 0 & \mathbb{1} \end{pmatrix}$, and all components are $2 \times 2$ block matrices. One can see that this $S$ is indeed a symplectic matrix, i.e., $SJ_3S^t=J_3$. Furthermore, we need to check whether this symplectic transformation gives the proper channel and the complementary channel of the thermal attenuator. The full transformation is written in terms of covariance matrices as
\begin{align}
S_{\text{th}}(\gamma_{\text{in}} \oplus \gamma_{\text{TMSV}})S_{\text{th}}^t&=\begin{pmatrix} X_0 & Z_0 & 0 \\ X_{c0} & Z_{c0} & 0 \\ 0 & 0 & \mathbb{1} \end{pmatrix}
\begin{pmatrix} \gamma_{\text{in}} & 0 & 0 \\ 0 & (2N+1)\mathbb{1} & 2\sqrt{N(N+1)}\mathbb{Z} \\ 0 & 2\sqrt{N(N+1)}\mathbb{Z} & (2N+1)\mathbb{1} \end{pmatrix}
\begin{pmatrix} X_0^t & X_{c0}^t & 0 \\ Z_0^t & Z_{c0}^t & 0 \\ 0 & 0 & \mathbb{1} \end{pmatrix} \nonumber \\
&=\begin{pmatrix} X_0 \gamma_{\text{in}} X_0^t+(2N+1)Z_0Z_0^t & X_0  \gamma_{\text{in}}X_{c0}^t  + (2N+1)Z_0 Z_{c0}^t & 2\sqrt{N(N+1)}Z_0 \mathbb{Z} \\
X_{c0} \gamma_{\text{in}}X_{0}^t+(2N+1)Z_0^t Z_{c0} & X_{c0}\gamma_{\text{in}}X_{c0}^t+(2N+1)Z_{c0}Z_{c0}^t & 2\sqrt{N(N+1)}Z_{c0}\mathbb{Z} \\
2\sqrt{N(N+1)}Z_0^t \mathbb{Z} & 2\sqrt{N(N+1)}Z_{c0}^t\mathbb{Z} & (2N+1)\mathbb{1} \end{pmatrix}.
\end{align}
If we trace out the environment modes, the covariance matrix after the channel action is $\gamma_{\text{out}}= X_0 \gamma_{\text{in}} X_0^t+(2N+1)Z_0Z_0^t=t\gamma_{\text{in}}+(2N+1)(1-t)\mathbb{1}$, as expected. If we trace out the input mode in order to obtain the output of the complementary channel, 
\begin{equation}
\gamma_{\text{com}}=\begin{pmatrix}X_{c0}\gamma_{\text{in}}X_{c0}^t+(2N+1)Z_{c0}Z_{c0}^t & 2\sqrt{N(N+1)}Z_{c0}\mathbb{Z} \\
2\sqrt{N(N+1)}Z_{c0}^t\mathbb{Z} & (2N+1)\mathbb{1} \end{pmatrix}=\begin{pmatrix}(1-t)\gamma_{\text{in}}+(2N+1)t\mathbb{1} & -2\sqrt{N(N+1)}\sqrt{t}\mathbb{Z} \\ -2\sqrt{N(N+1)}\sqrt{t}\mathbb{Z} & (2N+1)\mathbb{1} \end{pmatrix}.
\end{equation}
Here if we also trace out the ancillary mode used for purifying environment, the weak-complementary channel is obtained, i.e., $\gamma_{\text{wcom}}=(1-t)\gamma_{\text{in}}+(2N+1)t\mathbb{1}$.

From these results and the symplectic transformation of PPT channel given by
\begin{equation}
\small S_{\text{PPT}}=\left(
\begin{array}{cccccccc}
 \frac{a^2-1}{2 a} & 0 & \frac{a^2+1}{2 \sqrt{3} a} & 0 & \frac{a^2+1}{\sqrt{6} a} & 0 & 0 & 0 \\
 0 & -\frac{a^2-1}{2 a} & 0 & \frac{a^2+1}{2 \sqrt{3} a} & 0 & \frac{a^2+1}{\sqrt{6} a} & 0 & 0 \\
 -\frac{a^2+1}{2 \sqrt{3} a} & 0 & \frac{1}{6} \left(-a+2 b-\frac{2}{b}+\frac{1}{a}\right) & 0 & -\frac{(a+b) (a b-1)}{3 \sqrt{2} a b} & 0 & -\frac{b^2+1}{\sqrt{6} b} & 0 \\
 0 & -\frac{a^2+1}{2 \sqrt{3} a} & 0 & \frac{1}{6} \left(a-2 b+\frac{2}{b}-\frac{1}{a}\right) & 0 & \frac{(a+b) (a b-1)}{3 \sqrt{2} a b} & 0 & -\frac{b^2+1}{\sqrt{6} b} \\
 -\frac{a^2+1}{\sqrt{6} a} & 0 & -\frac{(a+b) (a b-1)}{3 \sqrt{2} a b} & 0 & \frac{1}{6} \left(-2 a+b-\frac{1}{b}+\frac{2}{a}\right) & 0 & \frac{b^2+1}{2 \sqrt{3} b} & 0 \\
 0 & -\frac{a^2+1}{\sqrt{6} a} & 0 & \frac{(a+b) (a b-1)}{3 \sqrt{2} a b} & 0 & \frac{1}{6} \left(2 a-b+\frac{1}{b}-\frac{2}{a}\right) & 0 & \frac{b^2+1}{2 \sqrt{3} b} \\
 0 & 0 & \frac{b^2+1}{\sqrt{6} b} & 0 & -\frac{b^2+1}{2 \sqrt{3} b} & 0 & -\frac{b^2-1}{2 b} & 0 \\
 0 & 0 & 0 & \frac{b^2+1}{\sqrt{6} b} & 0 & -\frac{b^2+1}{2 \sqrt{3} b} & 0 & \frac{b^2-1}{2 b} \\
\end{array}
\right) \nonumber
\end{equation}
\begin{align}
:=\begin{pmatrix}X_{\text{PPT}} & Z_{\text{PPT}} \\ X_{c,\text{PPT}} & Z_{c,\text{PPT}} \end{pmatrix}
\end{align}
where $a,b \in [1,\infty)$ and $X_{\text{PPT}}, Z_{\text{PPT}}, X_{c,\text{PPT}}, Z_{c,\text{PPT}}$ are $4 \times 4$ block matrices. Then we can finally obtain the symplectic transformation of the combined channel $\Phi_{\text{PPT}}\otimes \Phi_{\text{th}}$.
If we define $X=X_{\text{PPT}} \oplus X_0, Z=Z_{\text{PPT}} \oplus Z_0, X_c=X_{c,\text{PPT}} \oplus X_{c0}, Z_c=Z_{c,\text{PPT}} \oplus Z_{c0}$ as $6 \times 6$ matrices, the total symplectic transformation of the combined channel can be written as
\begin{equation}
S(\gamma_{\text{in}} \oplus \gamma_{\text{vac}} \oplus \gamma_{\text{TMSV}})S^t=\begin{pmatrix} X & Z & 0   \\ X_c & Z_c & 0  \\ 0 &0 & \mathbb{1}  \end{pmatrix} \begin{pmatrix}\gamma_{\text{in}} & 0 & 0 \\ 0 & \gamma_{\text{vac}} & 0 \\ 0 & 0 & \gamma_{\text{TMSV}} \end{pmatrix} \begin{pmatrix}  X^t & X^t_c & 0   \\ Z^t & Z^t_c & 0  \\ 0 &0 & \mathbb{1} \end{pmatrix} \nonumber
\end{equation}
\begin{equation}
=\begin{pmatrix} X\gamma_{\text{in}}X^t+Z_{\text{PPT}}Z^t_{\text{PPT}} \oplus Z_{\text{th}} \gamma_{\text{TMSV}} Z^t_{\text{th}} & (X\gamma_{\text{in}}X^t_c~,~0)+Z_{\text{PPT}}Z^t_{c,\text{PPT}} \oplus Z_{\text{th}}\gamma_{\text{TMSV}}Z^t_{c,\text{th}}
\\  \begin{pmatrix} X_c \gamma_{\text{in}}X^t \\ 0 \end{pmatrix}+Z_{c,\text{PPT}}Z^t_{\text{PPT}} \oplus Z_{c,\text{th}}\gamma_{\text{TMSV}} Z^t_{\text{th}} & \begin{pmatrix} X_c \gamma_{\text{in}} X^t_c & 0 \\ 0 & 0 \end{pmatrix}+Z_{c,\text{PPT}}Z^t_{c,\text{PPT}} \oplus Z_{c,\text{th}} \gamma_{\text{TMSV}} Z^t_{c,\text{th}} \end{pmatrix},
\end{equation}
where $\gamma_{\text{vac}}=\mathbb{1}_4$ and $\gamma_{\text{in}}$ is a channel input state with certain form as
\begin{equation}
\gamma_{\text {in}}=\left(
\begin{array}{cccccc}\label{input}
 \frac{x^4+1}{2 x^2} & 0 & 0 & 0 & \frac{\left(x^4-1\right) \left(y^2-1\right)}{4 x^2 y} & 0 \\
 0 & \frac{x^4+1}{2 x^2} & 0 & 0 & 0 & \frac{\left(x^4-1\right) \left(y^2-1\right)}{4 x^2 y} \\
 0 & 0 & \frac{z^4+1}{2 z^2} & 0 & \frac{\left(y^2+1\right) \left(z^4-1\right)}{4 y z^2} & 0 \\
 0 & 0 & 0 & \frac{z^4+1}{2 z^2} & 0 & -\frac{\left(y^2+1\right) \left(z^4-1\right)}{4 y z^2} \\
 \frac{\left(x^4-1\right) \left(y^2-1\right)}{4 x^2 y} & 0 & \frac{\left(y^2+1\right) \left(z^4-1\right)}{4 y z^2} & 0 & f(x,y,z) & 0 \\
 0 & \frac{\left(x^4-1\right) \left(y^2-1\right)}{4 x^2 y} & 0 & -\frac{\left(y^2+1\right) \left(z^4-1\right)}{4 y z^2} & 0 & f(x,y,z) \\
\end{array}
\right),
\end{equation}
where $f(x,y,z)=\frac{x^2 \left(y^2+1\right)^2 z^4+\left(x^4+1\right) \left(y^2-1\right)^2 z^2+x^2 \left(y^2+1\right)^2}{8 x^2 y^2 z^2}$, and $x,y,z \in [1, \infty)$ the squeezing parameters.
Consequently, the channel output and the complementary channel output are given by
\begin{align}
\gamma_{\text{out}}=&~X\gamma_{\text{in}}X^t+Z_{\text{PPT}}Z^t_{\text{PPT}} \oplus Z_{\text{th}} \gamma_{\text{TMSV}} Z^t_{\text{th}},
\\
\gamma_{\text{com}}=&\begin{pmatrix} X_c \gamma_{\text{in}} X^t_c & 0 \\ 0 & 0 \end{pmatrix}+Z_{c,\text{PPT}}Z^t_{c,\text{PPT}} \oplus Z_{c,\text{th}} \gamma_{\text{TMSV}} Z^t_{c,\text{th}}.
\end{align}

Next, we consider thermal amplifiers with amplifying parameter $G$, i.e., $\tau=G>1$. When $N=0$, the symplectic transformation is given by
\begin{equation}
S_{1}=\begin{pmatrix} \sqrt{G}\mathbb{1} & \sqrt{G-1}\mathbb{Z} \\ \sqrt{G-1}\mathbb{Z} & \sqrt{G}\mathbb{1} \end{pmatrix}.
\end{equation}
Let us set $X_1=\sqrt{G}\mathbb{1},~Z_1= \sqrt{G-1}\mathbb{Z},~X_{c1}=\sqrt{G-1}\mathbb{Z},~Z_{c1}=\sqrt{G}\mathbb{1}$. Then, by following same procedure for the thermal attenuator, we can obtain the symplectic transformation $S_{\text{am}}$ for $N>0$ as
\begin{equation}
S_{\text{am}}=\begin{pmatrix} X_{\text{am}} & Z_{\text{am}} \\ X_{c,\text{am}} & Z_{c,\text{am}} \end{pmatrix}=\begin{pmatrix} X_1 & Z_1 & 0 \\ X_{c1} & Z_{c1} & 0 \\ 0 & 0 & \mathbb{1} \end{pmatrix},
\end{equation}
where $X_{\text{am}}=X_1,~Z_{\text{am}}=\begin{pmatrix}Z_1 & 0\end{pmatrix},~X_{c,\text{am}}=\begin{pmatrix} X_{c1} \\ 0 \end{pmatrix},~Z_{c,\text{am}}=\begin{pmatrix} Z_{c1} & 0\\ 0 & \mathbb{1} \end{pmatrix}$, and all components represent $2 \times 2$ block matrices. Like the case of thermal attenuator, we need to check $S_{\text{am}}$ gives the proper channel and the complementary channel by looking at the full symplectic transformation as
\begin{align}
S_{\text{am}}(\gamma_{\text{in}} \oplus \gamma_{\text{TMSV}})S_{\text{am}}^t&=\begin{pmatrix} X_1 & Z_1 & 0 \\ X_{c1} & Z_{c1} & 0 \\ 0 & 0 & \mathbb{1} \end{pmatrix}
\begin{pmatrix} \gamma_{\text{in}} & 0 \\ 0 & \gamma_{\text{TMSV}}  \end{pmatrix}
\begin{pmatrix} X_1^t & X_{c1}^t & 0 \\ Z_1^t & Z_{c1}^t & 0 \\ 0 & 0 & \mathbb{1} \end{pmatrix} \nonumber \\
&=\begin{pmatrix}  X_{\text{am}}\gamma_{\text{in}} X_{\text{am}}^t+ Z_{\text{am}}\gamma_{\text{TMSV}} Z_{\text{am}}^t & X_{\text{am}}\gamma_{\text{in}} X_{\text{c,am}}^t+ Z_{\text{am}}\gamma_{\text{TMSV}} Z_{\text{c,am}}^t \\ X_{\text{c,am}}\gamma_{\text{in}} X_{\text{am}}^t+ Z_{\text{c,am}}\gamma_{\text{TMSV}} Z_{\text{am}}^t & X_{\text{c,am}}\gamma_{\text{in}} X_{\text{c,am}}^t+ Z_{\text{c,am}}\gamma_{\text{TMSV}} Z_{\text{c,am}}^t\end{pmatrix}.
\end{align}
After tracing out environment (system) modes, we get channel output (complementary channel output) written as
\begin{align}
\gamma_{\text{out}}=&~X_{\text{am}}\gamma_{\text{in}} X_{\text{am}}^t+ Z_{\text{am}}\gamma_{\text{TMSV}} Z_{\text{am}}^t=G\gamma_{\text{in}}+(2N+1)(G-1)\mathbb{1},
\\
\gamma_{\text{com}}=&X_{\text{c,am}}\gamma_{\text{in}} X_{\text{c,am}}^t+ Z_{\text{c,am}}\gamma_{\text{TMSV}} Z_{\text{c,am}}^t=\begin{pmatrix}(G-1)\mathbb{Z}\gamma_{\text{in}}\mathbb{Z}^t +(2N+1)G\mathbb{1} & 2\sqrt{N(N+1)}\sqrt{G}\mathbb{Z} \\ 2\sqrt{N(N+1)}\sqrt{G}\mathbb{Z} & (2N+1)\mathbb{1}\end{pmatrix}.
\end{align}
From these results, we can also construct the symplectic transformation of combined channel with PPT channel given by
\begin{equation}
S(\gamma_{\text{in}} \oplus \gamma_{\text{vac}} \oplus \gamma_{\text{TMSV}})S^t=\begin{pmatrix} X & Z & 0   \\ X_c & Z_c & 0  \\ 0 &0 & \mathbb{1}  \end{pmatrix} \begin{pmatrix}\gamma_{\text{in}} & 0 & 0 \\ 0 & \gamma_{\text{vac}} & 0 \\ 0 & 0 & \gamma_{\text{TMSV}} \end{pmatrix} \begin{pmatrix}  X^t & X^t_c & 0   \\ Z^t & Z^t_c & 0  \\ 0 &0 & \mathbb{1} \end{pmatrix} \nonumber
\end{equation}
\begin{equation}
=\begin{pmatrix} X\gamma_{\text{in}}X^t+Z_{\text{PPT}}Z^t_{\text{PPT}} \oplus Z_{\text{am}} \gamma_{\text{TMSV}} Z^t_{\text{am}} & (X\gamma_{\text{in}}X^t_c~,~0)+Z_{\text{PPT}}Z^t_{c,\text{PPT}} \oplus Z_{\text{am}}\gamma_{\text{TMSV}}Z^t_{c,\text{am}}
\\  \begin{pmatrix} X_c \gamma_{\text{in}}X^t \\ 0 \end{pmatrix}+Z_{c,\text{PPT}}Z^t_{\text{PPT}} \oplus Z_{c,\text{am}}\gamma_{\text{TMSV}} Z^t_{\text{am}} & \begin{pmatrix} X_c \gamma_{\text{in}} X^t_c & 0 \\ 0 & 0 \end{pmatrix}+Z_{c,\text{PPT}}Z^t_{c,\text{PPT}} \oplus Z_{c,\text{am}} \gamma_{\text{TMSV}} Z^t_{c,\text{am}} \end{pmatrix},
\end{equation}
where $X=X_{\text{PPT}} \oplus X_1, X_c=X_{c,\text{PPT}} \oplus X_{c1}$ as $6 \times 6$ matrices.

\label{appB}
\section{Non-activation of coherent information for entanglement-breaking channels with finite input energy}

Here, we generalize the non-activation property of coherent information known for finite-dimensional entanglement-breaking channels to infinite-dimensional entanglement-breaking channels with finite input energy. Our discussion is closely related to the one in Ref.\,\cite{Shirokov2006} on the Holevo $\chi$-function while applying the continuity result of the coherent information shown in Ref.\,\cite{Holevo2010}.

Let $D(\xx)$ denote the set of density operators acting on the Hilbert space $\xx$, and $T(\xx,\xx')$ be the set of superoperators $\Phi:D(\xx)\rightarrow D(\xx')$. We use curly letters for denoting Hilbert spaces and Roman letters for denoting the corresponding subsystems.

Let $\Phi\in T(\ma,\ma')$.
For finite-dimensional systems, mutual information of the channel and state is defined by
\ba
I(\rho,\Phi)=H(A)+H(A')-H(E)
\label{eq:mutual_finite}
\ea
where $E$ is the output system of the complementary channel.
On the other hand, for infinite-dimensional systems, this definition may be ill-defined since von Neumann entropy can be infinite.
To overcome this subtlety, Holevo and Shirokov introduced the following definiton.
\begin{defn}[\cite{Holevo2010}]
For $\Phi\in T(\ma,\ma')$ and $\rho\in D(\ma)$, mutual information with respect to $\rho$ and $\Phi$ is defined by
 \begin{eqnarray}
 I(\rho,\Phi)\equiv H((\mathbb{1}\otimes \Phi)\dm{\psi}||\rho\otimes\Phi(\rho))
 \end{eqnarray}
where $\dm{\psi}$ is a purification of $\rho$ and $H(\cdot||\cdot)$ is the relative entropy.
\end{defn}

Note that when $\dim\ma<\infty$ and $\dim\ma'<\infty$, this definition reduces to \eq{mutual_finite}.

Another important quantity, especially relevant to quantum capacity of a channel, is the coherent information.
For finite-dimensional systems, the coherent information of channel $\Phi$ and state $\rho$ is defined by
\ba
 I_c(\rho,\Phi) = H(A')-H(RA')
 \label{eq:coherent_finite}
\ea
where $R$ is the system purifying $\rho$.
For infinite-dimensional systems, this definition may be ill-defined even for the state $\rho$ with the finite von Neumann entropy since the entropy of the output state can be infinite.
To remedy this, the following definition was introduced.

\begin{defn}[\cite{Holevo2010}]
For $\Phi\in T(\ma,\ma')$ and $\rho\in D(\ma)$, coherent information with respect to $\rho$ and $\Phi$ is defined by
 \ba
 I_c(\rho,\Phi)\equiv I(\rho,\Phi)-H(\rho)
 \label{eq:coh_info_def}
 \ea
 where $H(\cdot)$ is the von Neumann entropy.
\end{defn}

When $H(\rho)<\infty$ and $H(\Phi(\rho))<\infty$, this definition reduces to \eq{coherent_finite}.
Note that when $H(\rho)$ is finite, $I_c(\rho,\Phi)$ is finite for arbitary $\Phi$ because
\ba
 I(\rho,\Phi) = H(\mathbb{1}\otimes\Phi(\dm{\psi})||\mathbb{1}\otimes\Phi(\rho\otimes\rho))\leq H(\dm{\psi}||\rho\otimes\rho)
\ea
where we used the monotonicity of the relative entropy.

We consider the following coherent information obtained as the supremum over all the input states with energy constraint.
\begin{defn}
Let $\ma$ be an infinite-dimensional Hilbert space corresponding to the bosonic system with the Hamiltonian $H=\sum_{n=0}^{\infty} n\dm{n}$. 
Let $\Phi\in T(\ma,\ma')$, and define $\tilde{D}_h(\ma)=\{\rho\in D(\ma)\,|\,\Tr[\rho H]<h\}$. Then, we define the coherent information with input energy constraint $h$ as
 \ba
 \tI_{c,h}(\Phi)\equiv \sup_{\rho\in \tilde{D}_h(\ma)}I_c(\rho,\Phi)
 \label{eq:coherent info energy constraint}
 \ea
\end{defn}

% Note that when $\dim\ma<\infty$, $\tD_h(\ma)$ becomes identical to $D(\ma)$ for sufficiently large $h$, which makes $\tI_{c,h}(\Phi)$ equivalent to the usual coherent information $I_c(\Phi)\equiv \sup_{\rho\in D(\ma)}I_c(\rho,\Phi)$.
% However, in this paper, we will use $\tI_c$ rather than $I_c$ for $\dim\ma<\infty$ as well for the notational simplicity.
For the case of finite input energy, the following important continuity property has been shown.
\begin{lem}[\cite{Holevo2010}]
 Let $\Phi\in T(\ma,\ma')$ and $\{\Phi_n\}$ be a sequence that strongly converges to $\Phi$. Then, for any sequence $\{\rho_n\}$ with $\forall n,\ \rho_n\in \tD_h(\ma)$ that converges to $\rho\in\tD_h(\ma)$, it holds that
 \ba
  \lim_{n\to\infty}I_c(\rho_n,\Phi_n) = I_c(\rho,\Phi)
 \ea
 \label{lemm:continuity}
for any $h<\infty$.
\end{lem}
For finite-dimensional channels consisting of an entanglement-breaking channel and an arbitrary channel, the following additivity result holds.
We include the proof of this result for completeness.

\begin{lem}[\cite{Watrous_app, Elton_app}]
 Let $\Phi_{\rm EB}\in T(\ma,\ma')$ be an entanglement-breaking channel and $\Psi\in T(\mb,\mb')$ be an arbitrary channel where $\dim\ma<\infty,\,\dim\ma'<\infty,\,\dim\mb<\infty,\,\dim\mb'<\infty$. Then,
 \ba
  I_c(\Phi_{\rm EB}\otimes\Psi)=I_c(\Psi)
 \ea
 \label{lemm:add_finite_EB}
\end{lem}
\begin{proof}
 Since the quantum capacity of any entanglement-breaking channel is zero due to the anti-degradablility of the entanglement-breaking channels and the non-cloning theorem, $I_c(\Phi_{\rm EB})=0$.
 $I_c(\Phi_{\rm EB}\otimes \Psi) \geq I_c(\Psi)$ is trivial, so it suffices to show $I_c(\Phi_{\rm EB}\otimes \Psi) \leq I_c(\Psi)$
 When input space and output space are finite-dimensional, the expression of coherent information of channel $\Phi_{\rm EB}\in T(\xx,\xx')$ and $\rho\in D(\xx)$ reduces to
 \ba
  I_c(\rho,\Phi_{\rm EB})= - H((\mathbb{1}\otimes\Phi_{\rm EB})\dm{\psi}_{RX}) + H(\Phi_{\rm EB}(\rho)) = -H(R|X')_{\mathbb{1}\otimes\Phi_{\rm EB}(\dm{\psi})}
 \ea
 where $\ket{\psi}\in \mathcal{R}\otimes \xx$ is a pure state purifying $\rho$, $R$ is a reference system for the purification, and $H(\cdot|\cdot)$ is the conditional entropy.

 Now, we consider $I_c(\rho,\Phi_{EB}\otimes\Psi)$ where $\rho\in D(\ma\otimes\mb)$.
 Let $\dm{\psi}\in  D(\mr\otimes\ma\otimes\mb)$ be a pure state purifying $\rho$, and define $\sigma = \mathbb{1}_{RB}\otimes\Phi_{\rm EB}(\dm{\psi})$.
 Since $\Phi_{\rm EB}$ is entanglement breaking, $\sigma$ can be written as $\sigma = \sum_y p_y\, \sigma^{A'}_y\otimes\sigma^{RB}_y$ for some probability distribution $\{p_y\}$ and pure states $\sigma_y^{A'}$, $\sigma_y^{RB}$.
 Define $\tau = \sum_y p_y \dm{y}_{R'}\otimes \sigma^{A'}_y\otimes\sigma^{RB}_y$ where we introduced another system $R'$.
 Then, we get
 \ba
  I_c(\rho,\Phi_{\rm EB}\otimes \Psi) &=& - H(R|A'B')_{\mathbb{1}_{RA'}\otimes\Psi(\sigma)}\\
  &\leq& - H(R|R'A'B')_{\mathbb{1}_{R'RA'}\otimes\Psi(\tau)}\\
  &=& -\left[H(R'RA'B') - H(R'A'B')\right]_{\mathbb{1}_{R'RA'}\otimes\Psi(\tau)}\\
  &=& - \left[H(RA'B'|R') - H(A'B'|R')\right]_{\mathbb{1}_{R'RA'}\otimes\Psi(\tau)}\\
  &=& -\sum_y p_y [H(RA'B')-H(A'B')]_{\sigma_y^{A'}\otimes [\mathbb{1}_R\otimes\Psi(\sigma_y^{RB})]}\\
  &=& -\sum_y p_y [H(RB')-H(B')]_{\mathbb{1}_R\otimes\Psi(\sigma_y^{RB})}\\
  &=& -\sum_y p_y H(R|B')_{\mathbb{1}_{R}\otimes\Psi(\sigma_y^{RB})}\\
  &=& \sum_y p_y I_c(\sigma_y^B,\Psi)\\
  &\leq& I_c(\Psi)
 \ea
where the first inequality is due to the strong subadditivity of the von Neumann entropy.
\end{proof}
In Ref.\,\cite{Shirokov2006}, the authors defined the Holevo capacity for infinite-dimensional channels and showed the additivity of the Holevo capacity of the channels consisting of an entanglement-breaking channel and an arbitrary channel.
Here, we basically apply their argument to the coherent information although there are some differences.
First difference is that the coherent information is continuous whereas Holevo $\chi$-function is only lower semicontinuous, which makes our analysis on the coherent information easier.
Second difference is that the $\chi$-function satisfies the following property
\ba
 \chi(\rho,\Phi_{\rm EB}\otimes\Psi) \leq \chi(\rho_A,\Phi_{\rm EB}) + \chi(\rho_B,\Psi),\ \forall \rho
\ea
for finite-dimensional channels while it is not clear whether the corresponding relation holds for the coherent information due to the lack of concavity with respect to the input state.
Thus, we need a slightly different analysis.

Let $\Phi\in T(\ma,\ma')$, and $P_n$ be a finite-rank projector acting on $\ma'$ such that $\lim_{n\to\infty}P_n = \mathbb{1}_{\ma'}$.
 Let $\ma_n'$ be a finite-dimensional subspace of $\ma'$ defined by $\ma_n'=P_n(\ma')$.
 Let us take another finite-dimensional subspace $\ma_n''\subset {\ma_n'}^{\perp}\subset\ma'$ and some pure state $\tau_n\in D(\ma_n'')$.
 Consider a sequence of channels $\Phi_n\in T(\ma,\ma'_n\oplus\ma_n'')$ defined by
 \ba
  \Phi_n(\cdot) = P_n\Phi(\cdot)P_n + \Tr[(\mathbb{1}_{\ma'}-P_n)\Phi(\cdot)]\tau_n.
  \label{eq:F_channel}
 \ea
Since $\lim_{n\to\infty}\Phi_n(\rho)=\Phi(\rho), \forall \rho\in D(\ma)$, the sequence $\{\Phi_n\}$ strongly converges to $\Phi$.
Note that $\Phi_n=\Pi_n\circ \Phi$ where $\Pi_n\in T(\ma',\ma'_n\oplus\ma''_n)$ is a channel defined by
\ba
  \Pi_n(\cdot) = P_n\cdot P_n + \Tr[(\mathbb{1}_{\ma'}-P_n)\cdot]\tau_n.
  \label{eq:Pi_channel}
\ea
Using these sequences of channels, we obtain the following. 
\begin{lem}
 Let $\Phi\in T(\ma,\ma')$ be a channel with $\dim\ma<\infty$, $\dim\ma'\leq\infty$, and $\Psi\in T(\mb,\mb')$ be a channel with $\dim\mb\leq\infty,\,\dim\mb'\leq\infty$.
 Define $\tD_{h_\mb}(\ma\otimes\mb)$ as the set of states whose reduced states acting on $\mb$ have the mean energy less than $h$.
 Then, for all $h<\infty$, if $\Phi_n$ defined by \eq{F_channel} satisfies
 \ba
  I_c(\rho,\Phi_n\otimes \Psi) \leq  I_c(\Phi_n)+\tI_{c,h}(\Psi),\ \forall \rho\in \tD_{h_\mb}(\ma\otimes\mb)
  \label{eq:ma_fin_assumption}
 \ea
 for all $n\in\mathbb{N}$, it holds that
 \ba
  I_c(\rho,\Phi\otimes \Psi) \leq  I_c(\Phi)+\tI_{c,h}(\Psi),\ \forall \rho\in \tD_{h_\mb}(\ma\otimes\mb).
 \ea
 \label{lemm:ma_fin}
\end{lem}
\begin{proof}
By the assumption \eq{ma_fin_assumption} and the compactness of $D(\ma)$, for any $n\in\mathbb{N}$ and $\rho\in\tD_{h_\mb}(\ma\otimes\mb)$, there exists $\sigma_n\in D(\ma)$ such that
 \ba
  I_c(\rho,\Phi_n\otimes\Psi) \leq I_c(\sigma_n,\Phi_n) + \tI_{c,h}(\Psi).
 \ea
Since $\Phi_n=\Pi_n\circ\Phi$ where $\Pi_n$ is defined by \eq{Pi_channel}, due to the monotonicity of the coherent information, we get $I_c(\sigma_n,\Phi_n)\leq I_c(\sigma_n,\Phi)$.
Combining the inequality $I_c(\sigma_n,\Phi)\leq I_c(\Phi)$, we get
 \ba
  I_c(\rho,\Phi_n\otimes\Psi) \leq I_c(\Phi) + \tI_{c,h}(\Psi).
 \ea
 Since $\lim_{n\to\infty}\Phi_n\otimes\Psi=\Phi\otimes\Psi$, the statement is obtained by using \lemm{continuity}.
\end{proof}

We next define subchannels, which are the channels with restricted input subspace. 
\begin{defn}
 The subchannel of $\Phi\in T(\ma,\ma')$ constrained on $\ma_0$, which is denoted by $\Phi_{\ma_0}$, is the channel in $T(\ma_0,\ma')$ where inputs are constrained to the set of states with support contained in a subspace $\ma_0\subset \ma$.
\end{defn}
Then, we obtain the following lemma.
\begin{lem}
Let $\Phi\in T(\ma,\ma')$ be a channel with $\dim\ma\leq\infty,\,\dim\ma'\leq\infty$ and $\Psi\in T(\mb,\mb')$ be a channel with $\dim\mb\leq\infty,\,\dim\mb'\leq\infty$.
Define $\tD_{h_\ma,h'_\mb}(\ma\otimes\mb)$ as the set of states whose reduced states acting on $\ma$ ($\mb$) is less than $h$ (h').
For any $h<\infty$ and $h'<\infty$, if it holds
\ba
 I_c(\rho,\Phi\otimes\Psi_{\ma_0\otimes\mb_0}) \leq \tI_{c,h}(\Phi_{\ma_0}) + \tI_{c,h'}(\Psi_{\mb_0}),\ \forall \rho\in \tD_{h_\ma,h'_\mb}(\ma\otimes\mb)
 \label{eq:all_infinite_assumption}
\ea
for any choice of $\ma_0\subset \ma$ and $\mb_0\subset\mb$ with $\dim\ma_0<\infty,\,\dim\mb_0<\infty$, then
\ba
 I_c(\rho,\Phi\otimes\Psi) \leq \tI_{c,h}(\Phi) + \tI_{c,h'}(\Psi),\ \forall \rho\in \tD_{h_\ma,h'_\mb}(\ma\otimes\mb).
\ea

 \label{lemm:all_infinite}
\end{lem}
\begin{proof}
 Consider the sequence of states
 \ba
  \rho_n = \left(\Tr[(P_n\otimes Q_n)\rho]\right)^{-1}(P_n\otimes Q_n)\,\rho\,(P_n\otimes Q_n)
 \ea
 where $P_n$ and $Q_n$ be finite-rank projectors acting on $\ma$ and $\mb$ such that $\lim_{n\to\infty}P_n=\mathbb{1}_{\ma}$ and $\lim_{n\to\infty}Q_n=\mathbb{1}_{\mb}$.
 Let $\Phi_{P_n}$ and $\Psi_{Q_n}$ be $P_n(\ma),Q_n(\mb)$-constrained channels.
 By assumption \eq{all_infinite_assumption}, for any $n\in\mathbb{N}$ and $\rho\in \tD_{h_\ma,h'_\mb}(\ma\otimes\mb)$, there exist $\sigma_n\in \tD_h(P_n(\ma))$ and $\tau_n\in \tD_{h'}(Q_n(\mb))$ such that
 \ba
  I_c(\rho_n,\Phi\otimes \Psi_{P_n(\ma)\otimes Q_n(\mb)}) \leq I_c(\sigma_n,\Phi_{P_n}) + I_c(\tau_n,\Psi_{Q_n})
 \ea
 where we used the compactness of $\tD_h(P_n(\ma))$ and $\tD_{h'}(Q_n(\mb))$.
 Since $\Phi_{P_n}$ and $\Psi_{Q_n}$ are just original channels with input restrictions, we get
 \ba
  I_c(\sigma_n,\Phi_{P_n}) &=& I_c(\sigma_n,\Phi)\leq \tI_{c,h}(\Phi)\\
  I_c(\tau_n,\Psi_{Q_n}) &=& I_c(\tau_n,\Psi)\leq \tI_{c,h'}(\Psi).
 \ea
 Since $\rho_n\rightarrow \rho$, $P_n\rightarrow \mathbb{1}_{\ma}$, $Q_n\rightarrow \mathbb{1}_{\mb}$, taking $n\to\infty$ and using \lemm{continuity}, we reach the statement.
\end{proof}

We finally reach our main result.
\begin{thm}
Let $\Phi\in T(\ma,\ma')$ be an entanglement-breaking channel with $\dim\ma\leq\infty,\,\dim\ma'\leq\infty$ and $\Psi\in T(\mb,\mb')$ be an arbitrary channel with $\dim\mb\leq\infty,\,\dim\mb'\leq\infty$. 
In a similar way to \eqref{eq:coherent info energy constraint}, define $\tI_{c,h_\ma,h'_\mb}(\Phi\otimes \Psi)$ as the coherent information obtained by taking the supremum over the states whose reduced states acting on $\ma$ ($\mb$) has the mean energy less than $h$ ($h'$). Then, for any $h<\infty$ and $h'<\infty$,
 \ba
  \tI_{c,h_\ma,h'_\mb}(\Phi\otimes \Psi) = \tI_{c,h'}(\Psi)
 \ea
\end{thm}
\begin{proof}
 Since the quantum capacity of any entanglement-breaking channel is zero due to the anti-degradablility of the entanglement-breaking channels and the no cloning theorem, $\tI_{c,h}(\Phi)=0$.
 $\tI_{c,h_\ma,h'_\mb}(\Phi\otimes \Psi) \geq \tI_{c,h'}(\Psi)$ is trivial, so it suffices to show $\tI_{c,h_\ma,h'_\mb}(\Phi\otimes \Psi) \leq \tI_{c,h'}(\Psi)$. To this end, we shall first show that
 \ba
  I_c(\rho,\Phi\otimes\Psi) \leq \tI_{c,h'}(\Psi),\ \forall \rho\in\tD_{h_\ma,h'_\mb}(\ma\otimes\mb).
  \label{eq:EB_state}
 \ea
 To show \eq{EB_state}, note that any subchannel of entanglement-breaking channel is also entanglement breaking. Thus, by virtue of \lemm{all_infinite}, it suffices to show that
 \ba
  I_c(\rho,\tilde{\Phi}\otimes\tilde{\Psi}) \leq \tI_{c,h'}(\tilde{\Psi}),\ \forall \rho\in\tD_{h_\ma,h'_\mb}(\tilde{\ma}\otimes\tilde{\mb})
  \label{eq:EB_FI}
 \ea
 for any entanglement-breaking channel $\tilde{\Phi}\in T(\tilde{\ma},\ma')$ with $\dim\tilde{\ma}<\infty$, $\dim\ma'\leq\infty$ and any channel $\tilde{\Psi}\in T(\tilde{\mb},\mb')$ with $\dim\tilde{\mb}<\infty$, $\dim\mb'\leq\infty$.
 This can be shown by using \lemm{ma_fin} twice. Let $\tilde{\Psi}'$ be a channel with input space as well as output space being finite-dimensional.
 Combining \lemm{add_finite_EB} with \lemm{ma_fin}, we get \eq{EB_FI} with $\tilde{\Psi}$ being replaced with $\tilde{\Psi}'$.
 We then use \lemm{ma_fin} again to promote $\tilde{\Psi}'$ to $\tilde{\Psi}$ to complete the proof of \eq{EB_FI}, which implies \eq{EB_state} by \lemm{all_infinite}.

\end{proof}

\end{appendices}


\begin{thebibliography}{10}
%
\bibitem{Holevo} A. S. Holevo,
Prob. Inf. Trans. {\bf 9}, 177 (1973).

\bibitem{Shor95} P. W. Shor,
Phys. Rev. A {\bf 52}, R2493 (1995).

\bibitem{Private} C. H. Bennett and G. Brassard,
Proc. IEEE Int. Conf. Computers Systems and Signal Processing, 175 (1984).

\bibitem{Hastings} M. B. Hastings,
Nat. Phys. {\bf 5}, 255 (2009).

\bibitem{Smith08} G. Smith, J. M. Renes, and J. A. Smolin,
Phys. Rev. Lett. {\bf 100}, 170502 (2008).

\bibitem{Li} K. Li, A. Winter, X. Zou, and G.-C. Guo,
Phys. Rev. Lett. {\bf 103}, 120501 (2009).

\bibitem{DiVin} D. P. DiVincenzo, P. W. Shor, and J. A. Smolin,
Phys. Rev. A {\bf 57}, 830 (1998).

\bibitem{Smith07} G. Smith and J. A. Smolin,
Phys. Rev. Lett. {\bf 98}, 030501 (2007).

\bibitem{Science} G. Smith and J. Yard,
Science {\bf 321}, 1812 (2008).

\bibitem{Smith} G. Smith, J. A. Smolin, and J. Yard,
Nat. Photonics {\bf 5}, 624 (2011).

\bibitem{RMP} C. Weedbrook, S. Pirandola, R. Garc\'{i}a-Patr\'{o}n, N. J. Cerf, T. C. Ralph, J. H. Shapiro, and S. Lloyd,
Rev. Mod. Phys. {\bf 84}, 621 (2012).

\bibitem{Lim18} Y. Lim and S. Lee,
Phys. Rev. A {\bf 98}, 012326 (2018).

\bibitem{book} J. Eisert and M. M. Wolf,
{\it Quantum Information with Continuous Variables of Atoms and Light}
(Imperial College Press, London, 2007).

\bibitem{Holevobook} A. S. Holevo,
{\it Quantum Systems, Channels, Information: A Mathematical Introduction}
( de Gruyter, Berlin, 2013).

\bibitem{Caves} C. M. Caves,
Phys. Rev. D {\bf 26}, 1817 (1982).

\bibitem{EB} A. S. Holevo,
Prob. Inf. Trans. {\bf 44}, 3 (2008).

\bibitem{Wildebook} M. Wilde,
{\it Quantum Information Theory}
(Cambridge University Press, Cambridge, 2017).

\bibitem{Watrous} J. Watrous,
{\it The Theory of Quantum Information}
(Cambridge University Press, Cambridge, 2018).

\bibitem{Elton} E. Y. Zhu,
Private communication.

\bibitem{Stine} W. F. Stinespring,
Proc. Am. Math. Soc. {\bf 6}, 211 (1955).

\bibitem{Devetak05} I. Devetak,
IEEE Trans. Inf. Theory, {\bf 51}, 44 (2005).

\bibitem{Hayden08} P. Hayden, P. W. Shor, and A. Winter,
Open Sys. Inf. Dyn., {\bf 15}, 71 (2008).

\bibitem{Adesso} G. Adesso, S. Ragy, and A. R. Lee,
Open Sys. Inf. Dyn., {\bf 21}, 1440001 (2014).

\bibitem{Serafini} A. Serafini,
{\it Quantum Continuous Variables: A Primer of Theoretical Methods}
(CRC Press, London, 2017).

\bibitem{Caruso} F. Caruso, J. Eisert, V. Giovannetti, and A. S. Holevo,
New J. Phys. {\bf 10}, 083030 (2008).

\bibitem{Holevo07} A. S. Holevo,
Prob. Inf. Trans. {\bf 43}, 1 (2007).

\bibitem{bound} R. F. Werner and M. M. Wolf,
Phys. Rev. Lett. {\bf 86}, 3658 (2001).

\bibitem{Holevo01} A. S. Holevo and R. F. Werner,
Phys. Rev. A {\bf 63}, 032312 (2001).

\bibitem{nohlower} K. Noh and L. Jiang,
arXiv:1811.06989.

\bibitem{upper} M. Rosati, A. Mari, and V. Giovannetti,
Nat. Commun. {\bf 9}, 4339 (2018).

\bibitem{noh} K. Noh, V. V. Albert, and L. Jiang,
IEEE Trans. Inf. Theory, doi: 10.1109/TIT.2018.2873764.

\bibitem{plob} S. Pirandola, R. Laurenza, C. Ottaviani, and L. Banchi,
Nat. Commun. {\bf 8}, 15043 (2017).

\bibitem{thmode} G. Adesso, A. Serafini, and F. Illuminati,
Phys. Rev. A {\bf 73}, 032345 (2006).



\bibitem{Kamil} K. Br\'{a}dler,
J. Math. Phys. A {\bf 48}, 125301 (2015).

\bibitem{Wolf13} D. Lercher, G. Giedke, and M. M. Wolf,
New J. Phys. {\bf 15}, 123003 (2013).


%


\end{thebibliography}

\begin{thebibliography}{10}
%
\bibitem{Holevo2010} A. S. Holevo and M. E. Shirokov,
Prob. Inf. Tran. {\bf 46}, 201 (2010).

\bibitem{Wolf2007} M. M. Wolf, D. P\'{e}rez-Garc\'{i}a, and G. Giedke,
Phys. Rev. Lett. {\bf 98}, 130501 (2007).

\bibitem{Shirokov2006} M. E. Shirokov,
Comm. Math. Phys. {\bf 262}, 137 (2006).

\bibitem{Watrous_app} J. Watrous,
{\it The Theory of Quantum Information}
(Cambridge University Press, Cambridge, 2018).

\bibitem{Elton_app} E. Y. Zhu,
Private communication.

\end{thebibliography}
\end{document}